\documentclass{article}
\usepackage{fullpage}

\usepackage[utf8]{inputenc} 
\usepackage[T1]{fontenc}    
\usepackage{hyperref}       
\usepackage{url}            
\usepackage{booktabs}       
\usepackage{amsfonts}       
\usepackage{nicefrac}       
\usepackage{microtype}      
\usepackage{xcolor}         

\usepackage{hyperref}
\hypersetup{pdftex,colorlinks=true,allcolors=blue} 
\usepackage{hypcap} 

\usepackage{amsmath, amsthm, amssymb}
\usepackage{verbatim}
\usepackage{graphicx}
\usepackage{array}
\usepackage{mathrsfs}
\usepackage{comment}
\usepackage{enumitem} 
\usepackage{lmodern}

\newtheorem{lemma}{Lemma}
\newtheorem{proposition}{Proposition}

\newcommand{\R}{\mathbb{R}}
\newcommand{\E}{\mathbb{E}}

\DeclareMathOperator*{\argmin}{arg\,min}

\def\sU{{\mathsf U}}

\def\sX{{\mathsf X}}
\def\sY{{\mathsf Y}}
\def\sZ{{\mathsf Z}}

\def\deq{\triangleq}

\newcommand\blfootnote[1]{%
	\begingroup
	\renewcommand\thefootnote{}\footnote{#1}%
	\addtocounter{footnote}{-1}%
	\endgroup
}

\title{
One if by Land, Two if by Sea, Three if by Four Seas, and More to Come --- Values of Perception, Prediction, Communication, and Common Sense  in Decision Making}
\author{Aolin Xu} 
\date{}

\begin{document}
	
	\maketitle
	\blfootnote{xuaolin@gmail.com}

\begin{abstract}
This work aims to rigorously define the values of perception, prediction, communication, and common sense in decision making.
  The defined quantities are decision-theoretic, but have information-theoretic analogues, e.g., they share some simple but key mathematical properties with Shannon entropy and mutual information, and can reduce to these quantities in particular settings.
 One interesting observation is that, the value of perception without prediction can be negative, while the value of perception together with prediction and the value of prediction alone are always nonnegative.
 The defined quantities suggest answers to practical questions arising in the design of autonomous decision-making systems. Example questions include: Do we need to observe and predict the behavior of a particular agent? How important is it? What is the best order to observe and predict the agents?
 The defined quantities may also provide insights to cognitive science and neural science, toward the understanding of how natural decision makers make use of information gained from different sources and operations.
\end{abstract}
\section{Introduction}

\subsection{Motivation}
The design of autonomous decision-making systems motivates us to formulate and answer the following questions.
\begin{itemize}[leftmargin=*]
\item
What is the value of observing the world, including any agent in it, when making decisions, compared to not observing the world at all?
\item
When only part of the world can be observed but the rest can be predicted, what is the value of predicting the unobserved part of the world?
\item 
When there are multiple agents in the world to be considered in decision making,
how important is it for the decision maker to observe and predict the behavior of each agent?  
\item
What is the best order to perform perception and prediction on these agents?
\end{itemize} 
Answering the first two questions can help us quantify the importance of the perception and prediction subsystems in a decision-making system.
The latter two questions are of particular interest when the limited computational resource is a concern in the design of the decision-making system.
In this work, we attempt to find a rigorous path to formalize the above questions and derive answers to them.
We define the values of key operations for decision making, including perception and prediction, and extend to communication and reasoning with common sense; we also extend these definitions to decision-making problems with multiple agents.
 One interesting observation is that, the value of perception without prediction can be negative, while the value of perception together with prediction and the value of prediction alone are always nonnegative.
We analyze the mathematical properties of these values and draw connections to information-theoretic measures.
The defined quantities may also provide insights to cognitive and neural sciences, toward the understanding of how natural decision makers make use of information gained from different sources and operations.

\subsection{Related works}\label{appd:survey}
\noindent\textbf{Value of information}

The formal study of the value of information in decision making can be traced back at least to Blackwell’s work on the comparison of experiments through the reduction of expected loss in the binary hypothesis setup \cite{blackwell1951comparison, blackwell1953equivalent}. 
Lindley \cite{lindley1956measure} applied Shannon's notion of mutual information to measure the informativeness of experiments in the setup of general Bayesian parameter estimation. 
DeGroot \cite{DeGroot62} considered general Bayesian decision making and generalized the notions of Shannon entropy and mutual information through arbitrary loss functions.
Further theoretical developments of this concept include Howard \cite{Howard_voi}, Gr\"unwald and Dawid \cite{max_gen_ent}, and Xu and Raginsky \cite{MER22} with applications in economics, robust decision making, and machine learning.

This work can be viewed as an extension of the above line of research. We further generalize the notion of entropy and mutual information through the analysis of the value, or the reduction of minimum achievable expected loss, of various operations in decision making. In those earlier works, the ``value of information'' can be viewed as the special case of the value of perception together with prediction as defined in this work. 
We also extend the discussion to cover common sense and communication as well. 
There is no dedicated study in the literature on the values of the fine-grained and diverse operations considered in this work yet. 
The decision-making problem considered here is also more general, due to the more general form of the loss function we use.
The loss function we use takes both the observable and unobservable states of the world, as well as the action into consideration, which enables the setup to cover a wide range of decision-making, estimation, and even game-playing problems.

\noindent\textbf{Information consumption in perception-prediction-action loop}

A potentially striking insight from the definitions proposed in this work is that the narrow value of perception without prediction can be negative, whereas the value of perception together with prediction and the value of prediction alone are provably nonnegative. In other words, just having sensory data without properly incorporating its significance can, in some cases, worsen decision quality. A real-world myopic driving example is provided to illustrate this fact. 
This fact resonate the concept of rational inattention proposed by Sim \cite{sims2003implications, sims2006rational}: if an agent cannot process the information, potentially due to limited ``prediction'' capability, then acquiring raw data might not be worthwhile.
Empirical studies of human perception and decision making also reflect the importance of the ability to consume the information than merely acquiring the information \cite{BREHMER1992211, MALONEY20102362, Herd21, Morelli22, voi_human_ai}.
Recently, common sense reasoning in perception, prediction, and planning also becomes a trending topic in machine learning and AI research \cite{common_sense_cvpr23, AIAlignment}. Further, the formulation of artificial general intelligence may also be carried out in a value-communication based framework \cite{CUV2025}.  Gaining a principled understanding of the importance of values of information in decision making helps to advance the related research with a solid foundation.

\noindent\textbf{Agent importance and prioritization}

The identification of importance of agents in machine learning and decision making become an important problem as more autonomous or AI systems get more involved in people's daily life. 
This problem is more critical when the computation power becomes a bottleneck for simultaneously tracking or paying attention to all agents involved in decision making.
Example works in this topic include \cite{agent_priority19,goal_important19,Zhang2020InteractionGF,important_semi22} in the area of autonomous driving. 
Those works are mostly based on machine learning methods which require human labels on the importance of the agents, or based on rules from human knowledge. The results in this work can potentially be used to provide ground truth labels of agent importance to be used by the learning-based methods, e.g. the data-driven approach proposed in \cite{rank2tell}.

\section{Definitions of values of different operations}\label{sec:formulation}
\subsection{A generic decision-making problem}
This study considers a generic decision-making problem, where a decision maker needs to take an action $u$ from action space $\sU$ in order to minimize the expectation $\E[\ell(X,Y,u)]$ of a loss function $\ell:\sX\times\sY\times\sU\rightarrow\R$, where $X$ is a random observable state, $Y$ is a random unobservable state, and the expectation is taken with respect to the joint distribution $P_{X,Y}$.
$X$ being observable means that the action can be taken as a function of $X$ through a policy $\psi:\sX\rightarrow\sU$, such that $U = \psi(X)$.
The minimum achievable expected loss is defined as the risk of the decision.
This formulation encompasses several types of problems, e.g.
\begin{itemize}[leftmargin=*]
\item
One-step or open-loop stochastic control, where the loss depends on both the observed and unobserved states. This includes the settings where the loss only depends on $(X,U)$ or $(Y,U)$ as special cases.
\item 
Bayesian statistical estimation, including both regression and classification, where $Y$ needs to be estimated based on $X$. Usually the loss is considered to be only dependent on $(Y,U)$ in this case. An advantage of considering both $X$ and $Y$ in the loss function is that it makes the loss contextual, which can better model the situations where the quality of an estimate depends on its accuracy as well as the observable context it is made in.
\item 
Bayesian game play \cite{harsanyi1967i} between ego player and an agent, where $X$ includes the ego player's own state and the observable state of the agent, $Y$ represents the unknown action to be taken by the agent that statistically depends on $X$, and the loss depends on both players' states and actions.
\end{itemize}
For this decision-making problem, how well the decision can be made depends on what information or knowledge the decision maker uses when coming up with the action. Possible knowledge that can be used includes the marginal distributions $P_X$ and $P_Y$, the joint distribution $P_{X,Y}$, the observable state $X$, and even the unobservable state $Y$. There are also different ways of using $P_{X,Y}$, depending on whether it is used together with $X$ or not.
Gaining and making use of these knowledge usually requires different operations, such as common sense reasoning, perception, prediction, and even communication. In general, the more knowledge is gained, the better decision can be made.
In what follows, we present a rigorous way to define the operations for gaining different knowledge, and quantify the values of these operations when the gained knowledge is used for decision making.

\subsection{Value of common sense}
When neither $X$ nor $Y$ can be used for decision making, knowing their joint distribution can still be helpful as it encodes their statistical dependence, which is usually gained from common sense.
We can thus quantify the value of common sense as the reduction of the minimum achievable expected loss, or the risk, when using $P_{X,Y}$ in decision making, compared with when ignoring the statistical dependence between $X$ and $Y$ by using only the product of their marginals $P_X P_Y$. Formally, this value can be defined as
	\begin{align}\label{eq:vocomsen_def}
		 \E[\ell(X,Y,\bar u)] - \E[\ell(X,Y,u^*)] 
	\end{align}
    where
\begin{align}\label{eq:ubar_def}
\bar u\deq\argmin_{u\in\sU} \E[\ell(X,\bar Y,u)]
\end{align}
with $P_{X,\bar Y} = P_X P_Y$ is the best action to take when $X$ and $Y$ are treated to be independent from each other; and
\begin{align}\label{eq:u*_def}
    u^*\deq\argmin_{u\in\sU} \E[\ell(X,Y,u)]
\end{align}
is the best action to take when using the true joint distribution of $X$ and $Y$.
From the definition of $u^*$, we know that the value of common sense is always nonnegative. 

 \subsection{Value of perception}
The observable state $X$ can be used in decision making only when it is perceived. In a narrow sense, perception is the operation to get $X$ ready for use by a policy, while not taking the dependence of $Y$ on $X$ into account. We can define the value of this narrow sense of perception as the reduction of risk when using $X$ in decision making without predicting $Y$ based on $X$, compared with only using $P_{X,Y}$. Formally, it can be defined as
\begin{align}\label{eq:vopercep_def}
		 \E[\ell(X,Y,u^*)] - \E[\ell(X,Y,\bar\psi(X))] 
	\end{align}
where
\begin{align}\label{eq:psibar_def}
\bar\psi\deq\argmin_{\psi:\sX\rightarrow\sU} \E[\ell(X,\bar Y,\psi(X))]
\end{align}
is the best policy when using $X$ but treating $Y$ to be independent of $X$ in decision making.
Given any realization of $X=x$, $\bar\psi(x)$ can be realized as 
\begin{align}
\bar\psi(x) = \argmin_{u\in\sU}\E[\ell(x,Y,u)] .
\end{align}
Clearly, $\bar\psi$ only takes $X$ and the marginal distribution $P_Y$ into account when coming up with the best action to take, thus does not predict $Y$ based on $X$.
There is no definite order between $\E[\ell(X,Y,u^*)]$ and $\E[\ell(X,Y,\bar\psi(X))]$, meaning that the value of perception defined above could be negative for certain $P_{X,Y}$ and $\ell$. 
One such example is where $\sX = \sY = \sU = \{0,1\}$, with the loss function 


\begin{center}
\begin{tabular}
{ |c|c|c|c|c|c|c|c|c| }
\hline
 $(x,y,u)$ & 000 & 010 & 100 & 110 & 001 & 011 & 101 & 111\\ 
 \hline
 $\ell$ & 0 & 0 & 0 & 0 & -1 & 1 & 1 & 1 \\  
 \hline
\end{tabular}
\end{center}


the joint distribution $P_{X,Y}$ 
\begin{center}
\begin{tabular}{ |c|c|c|c|c| }
\hline
 ($x$, $y$) & 00 & 01 & 10 & 11  \\ 
 \hline
 $P_{X,Y}$ & 0.04 & 0.06 & 0.81 & 0.09  \\  
 \hline
\end{tabular}
\end{center}
and hence the conditional distribution $P_{Y|X}$ and the marginals 
\begin{center}
\begin{tabular}{ |c|c|c|c|c| }
\hline
 ($x$, $y$) & 00 & 01 & 10 & 11  \\ 
 \hline
 $P_{Y|X}$ & 0.4 & 0.6 & 0.9 & 0.1  \\  
 \hline
\end{tabular}
\quad
\begin{tabular}{ |c|c|c| }
\hline
 $x$ & 0 & 1   \\ 
 \hline
 $P_{X}$ & 0.1 & 0.9   \\  
 \hline
\end{tabular}
\quad
\begin{tabular}{ |c|c|c| }
\hline
 $y$ & 0 & 1  \\ 
 \hline
 $P_{Y}$ & 0.85 & 0.15  \\  
 \hline
\end{tabular}
\end{center}
For this example, $u^*=0$ and $\E[\ell(X,Y,u^*)]=0$, while $\bar\psi(0)=1$, $\bar\psi(1)=0$, and $\E[\ell(X,Y,\bar\psi(X))] = 0.02$, resulting in a negative value of perception. 

\begin{figure}[t]
	\centering
\includegraphics[scale = 0.3]{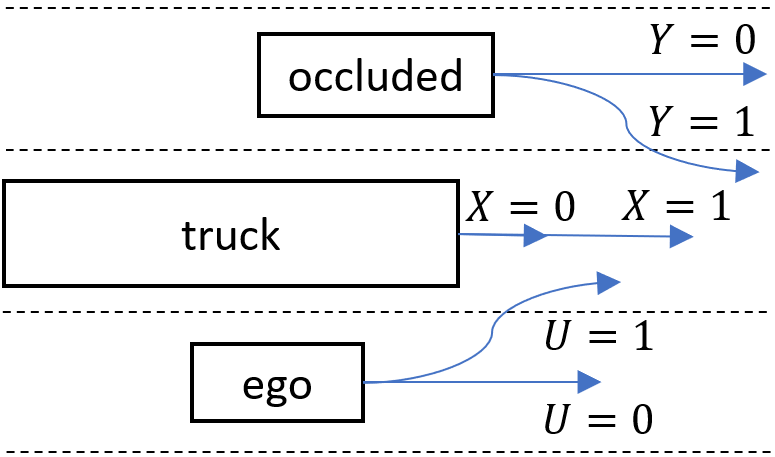}
	\caption{Real-world example where the value of perception without prediction can be negative: observing a truck slowing down without considering a potentially cutting-in vehicle occluded by the truck may result in a dangerous lane change of ego vehicle.}
	\label{fig:voperception_ex}
\end{figure}

The above numerical example may be used to model the following real-world driving situation illustrated in Fig.~\ref{fig:voperception_ex}: the ego vehicle need to decide whether to perform a left lane change, based on the observation of the behavior of a truck in the adjacent lane, which may be influenced by another vehicle in the truck's neighboring lane occluded by the truck.  The occluded vehicle's behavior can be $Y=0$ for lane keep or $Y=1$ for right lane change; the truck's behavior can be $X=0$ for braking or $X=1$ for not braking; while the ego vehicle's action can be $U=0$ for lane keep or $U=1$ for left lane change. The marginal distributions $P_X$ and $P_Y$ indicate that it is unlikely a-priori for the truck to brake or for the occluded vehicle to make a right lane change; while the conditional distribution $P_{Y|X}$ indicates that when the truck brakes, it is likely that some occluded vehicle is making a right lane change. The loss function $\ell=0$ indicates safe driving, $\ell=1$ indicates collision or unsafe driving, and $\ell=-1$ at $(x,y,u)=(0,0,1)$ indicates the incentive for a safe left lane change.
Without observing the truck's behavior, common sense would keep ego vehicle safe by preventing it from a dangerous left lane change. On the other hand, observing truck's braking but not taking its implication on the occluded vehicle's behavior into consideration, a left lane change would lead the ego vehicle to a dangerous situation. 

In general, the potentially negative value of perception explains that, it may be less risky for a decision maker to take a conservative action when they cannot clearly perceive the  situation, than acting with limited consciousness by only reacting to immediately perceivable factors but ignoring the interactions between them and the invisible ones.
The full potential of perception needs to be realized by the additional operation of prediction, as discussed in the next two subsections.

\subsection{ Value of prediction}
The value of prediction can be defined as the reduction of risk when making use of $X$ as well as predicting $Y$ based on it in decision making, compared with just using $X$ and treating $Y$ as independent of $X$. Formally, it can be defined as
\begin{align}\label{eq:vopred_def}
 \E[\ell(X,Y,\bar\psi(X))] - \E[\ell(X,Y,\psi^*(X))]
	\end{align}
    where
\begin{align}\label{eq:psi*_def}
\psi^*\deq\argmin_{\psi:\sX\rightarrow\sU} \E[\ell(X, Y,\psi(X))]
   \end{align}
is the best policy when using $X$ and the true joint distribution $P_{X,Y}$ in decision making.
Given $X=x$, $\psi^*(x)$ can be realized as 
\begin{align}
\psi^*(x) = \argmin_{u\in\sU}\E[\ell(x,Y,u)|X=x] . \label{eq:pred_psi}
\end{align}
In other words, $\psi^*$ takes both $X$ and the posterior $P_{Y|X}$ into account when coming up with the best action to take.
Computationally, the operation of predicting $Y$ based on $X$ means the evaluation of the posterior $P_{Y|X=x}$ for a given $x$, which is then used in the computation of \eqref{eq:pred_psi}.
From the definition of $\psi^*$, we know that the value of prediction is always nonnegative.

\subsection{Value of perception together with prediction}
We can also define the value of perception together with prediction as the reduction of risk when using $X$ and predicting $Y$ in decision making, compared with only using $P_{X,Y}$:
\begin{align}\label{eq:vo_percep_pred_def}
 \E[\ell(X,Y,u^*)] - \E[\ell(X,Y,\psi^*(X))]
	\end{align}
which simply equals the sum of the value of perception and the value of prediction defined above.
From the definition of $\psi^*$, we know that the value of perception together with prediction is always nonnegative.
When the loss only depends on $(Y,U)$, the value of perception together with prediction defined above becomes the classical value of information in decision making as defined in the early literature, e.g. in \cite{DeGroot62}
.  
\subsection{ Value of communication}
When the decision maker not only uses observable $X$, but also can access $Y$ through some form of communication, the risk in decision making can be further reduced. This reduction can be defined as the value of communication:
\begin{align}\label{eq:vopcomu_def}
 \E[\ell(X,Y,\psi^*(X))] - \E[\ell(X,Y,\psi^{**}(X,Y))]
	\end{align}
where
\begin{align}\label{eq:psi**_def}
\psi^{**}\deq\argmin_{\psi:\sX\times\sY\rightarrow\sU} \E[\ell(X, Y,\psi(X,Y))]
\end{align}
is the best communication-enabled policy that uses both $X$ and $Y$ as well as $P_{X,Y}$ in decision making. 
Given $X=x$ and $Y=y$, $\psi^{**}(x,y)$ can be realized as 
\begin{align}
\psi^{**}(x,y) = \argmin_{u\in\sU}\ell(x,y,u) ,
\end{align}
which does not rely on $P_{X,Y}$.
Since a suboptimal communication-enabled policy defined as $\psi(x,y) = \psi^*(x)$ that ignores $Y$ can achieve $\E[\ell(X,Y,\psi^*(X))]$, we know that the value of communication is always nonnegative.

\section{Information-theoretic realizations and interpretations}
\begin{figure*}[t]
	\centering
	\includegraphics[scale = 0.3]{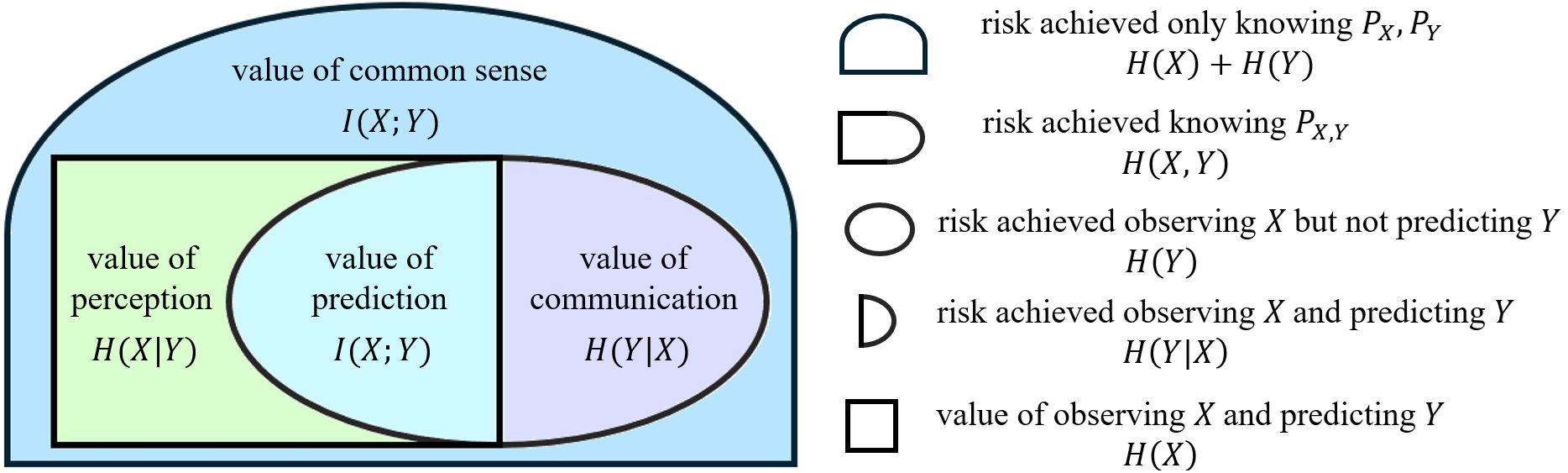}
	\caption{Values of common sense, perception, prediction, and communication, as well as their information-theoretic analogues, are shown as colored regions. Interpretations of some combined regions are shown on the right.}
	\label{fig:vop_diag}
\end{figure*}
The values of different operations defined above have information-theoretic realizations as commonly used information measures. These information measures in turn provide information-theoretic interpretations or analogues of the defined values.
Consider the case where $X$ and $Y$ are jointly distributed discrete random variables, the action space is the set of joint distributions on $\sX\times\sY$, and the loss function is $\ell(x,y,u) = -\log u(x,y)$, that is, the negative log-likelihood of $(x,y)$ under distribution $u$.
In this case, the best actions and policies defined above become
$
\bar u = P_X P_Y
$,
$
u^* = P_{X,Y}
$,
$
    \bar\psi(x) = \delta_x P_Y
$,
$
    \psi^*(x) = \delta_x P_{Y|X=x}
$,
$
    \psi^{**}(x,y) = \delta_x \delta_y
$,
where $\delta_x$ and $\delta_y$ are one-hot distributions at $x$ and $y$ respectively.
The various risks, or minimum achievable expected losses, become
\begin{equation}
    \begin{aligned}
&
 \E[\ell(X,Y,\bar u)] = \E[-\log (P_X(X) P_Y(Y))] 
 = H(X) + H(Y) \\
& \E[\ell(X,Y,u^*))] = \E[-\log P_{X,Y}(X,Y)] 
 = H(X,Y) \\
& \E[\ell(X,Y,\bar\psi(X))] = \E[-\log(\delta_{X}(X)P_{Y}(Y))] 
 = H(Y) \\
& \E[\ell(X,Y,\psi^*(X))] = \E[-\log(\delta_{X}(X)P_{Y|X}(Y|X))] 
 = H(Y|X) \\
& \E[\ell(X,Y,\psi^{**}(X,Y))] = \E[-\log(\delta_{X}(X)\delta_{Y}(Y))] 
 = 0 
\end{aligned}
\end{equation}
The various values defined above become
\begin{itemize}[leftmargin=*]
\item
value of common sense: \quad
$
  H(X) + H(Y) - H(X,Y) 
 = I(X;Y)
$
\item
value of perception: \quad
$
  H(X,Y) - H(Y)
 = H(X|Y)
$
\item 
value of prediction: \quad
$  
H(Y) - H(Y|X)
 = I(X;Y)
$
\item 
value of perception together with prediction: \quad
$
  H(X,Y) - H(Y|X)
 = H(X)
$
\item 
 value of communication: \quad
$
  H(Y|X) - 0
 = H(Y|X)
$
\end{itemize}
It is interesting to see that in this information-theoretic setting, the values of common sense and prediction both are the mutual information between $X$ and $Y$.
It is also notable that the uncertainty of $X$ can be fully removed in this decision-making problem only when both observing $X$ and predicting $Y$ based on $X$, as
\begin{align}
    H(X,Y) - H(Y|X) &= H(X|Y) + I(X;Y) 
    = H(X) .
\end{align}

The defined quantities and their information-theoretic realizations are illustrated in Fig.~\ref{fig:vop_diag}.
The intimate connections between the defined values and the information measures are further studied through their functional properties in the next section.

\section{Functional properties of defined values}\label{appd:func_prop}
We have seen that the defined values are always nonnegative except for the value of perception without prediction.
In this section we further study some basic functional properties of the defined values, mainly concavity and convexity in the underlying distributions. 
These properties are inherited as key properties by their information-theoretic realizations \cite{it_book_pw}. 

The results are built on four lemmas.
The first three are about properties of the risks achieved with different knowledge, which are useful in proving properties of the value of perception together with prediction and the value of communication stated in Propositions~\ref{prop:vo_percep_pred} and~\ref{prop:vo_commun}.

\begin{lemma}\label{lm:R(X,Y)}
$\E[\ell(X,Y,u^*)]$ is concave in the joint distribution $P_{X,Y}$.
\end{lemma}
\begin{proof}
By noting that $\E[\ell(X,Y,u)]$ is linear in the joint distribution $P_{X,Y}$ for each $u\in\sU$, we know that $ \E[\ell(X,Y,u^*] = \min_{u\in\sU}\E[\ell(X,Y,u)]$ is the minimum of a set of linear functions of $P_{X,Y}$. This implies that $\E[\ell(X,Y,u^*)]$ is concave in $P_{X,Y}$.   
\end{proof}

\begin{lemma}\label{lm:R(X,Y|X)}
$\E[\ell(X,Y,\psi^*(X))]$ has the following properties:
\begin{enumerate}[leftmargin=*]
    \item 
    For a fixed $P_{X}$, it is concave in $P_{Y|X}$.
    \item 
    For a fixed $P_{Y|X}$, it is linear in $P_X$.
\end{enumerate}
\end{lemma}
\begin{proof}
We have
\begin{align}
&\E[\ell(X,Y,\psi^*(X))] \nonumber \\
=& \sum_{x\in\sX}P_X(x) \E[\ell(x,Y,\psi^*(x))|X=x] \\
=& \sum_{x\in\sX} P_X(x) \min_{u\in\sU} \E[\ell(x,Y,u)|X=x] .
\end{align}
As $\E[\ell(x,Y,u)|X=x]$ is linear in $P_{Y|X=x}$, we know that $\min_{u\in\sU} \E[\ell(x,Y,u)|X=x]$ is concave in $P_{Y|X=x}$ for each $x\in\sX$.
With a fixed $P_X$, $\E[\ell(X,Y,\psi^*(X))]$ is a nonnegative linear combination of concave functions in $P_{Y|X}$, thus is still concave in $P_{Y|X}$, which proves the first claim.
The second claim simply follows from the fact that $\min_{u\in\sU} \E[\ell(x,Y,u)|X=x]$ does not depend on $P_X$ when $P_{Y|X=x}$ is fixed.
\end{proof}

\begin{lemma}\label{lm:R(X,Y|X,Y)}
$\E[\ell(X,Y,\psi^{**}(X,Y))]$ is linear in the joint distribution $P_{X,Y}$.
\end{lemma}
\begin{proof}\label{appd:pf_lm_R(X,Y|X,Y)}
We have
\begin{align}
& \E[\ell(X,Y,\psi^{**}(X,Y))] \nonumber \\
=& \sum_{(x,y)\in\sX\times\sY}P_{X,Y}(x,y) \ell(x,y,\psi^{**}(x,y)) \\
=& \sum_{(x,y)\in\sX\times\sY}P_{X,Y}(x,y) \min_{u\in\sU} \ell(x,y,u) .
\end{align}   
The claim follows from the fact that $\min_{u\in\sU} \ell(x,y,u)$ does not depend on $P_{X,Y}$.
\end{proof}

The next lemma from \cite{max_gen_ent} is useful in proving the properties of the values of common sense and prediction stated in Proposition~\ref{prop:vo_coms_pred}.
\begin{lemma}\label{lm:D_ell}
Given a loss function $\ell:\sZ\times\sU\rightarrow\R$ and two distributions $P$ and $Q$ on $\sZ$, a decision-theoretic divergence between $P$ and $Q$ with respect to $\ell$ that generalizes the Kullback-Liebler divergence can be defined as
\begin{align}
D_\ell(P,Q) = \E_P[\ell(Z,u_Q)] - \E_P[\ell(Z,u_P)]
\end{align}
where $u_P = \argmin_{u\in\sU}\E_P[\ell(Z,u)]$ and $u_Q = \argmin_{u\in\sU}\E_Q[\ell(Z,u)]$.
This divergence is convex in $P$ when $Q$ is fixed.
\end{lemma}

Now we can state the major functional properties in the following three propositions.
\begin{proposition}\label{prop:vo_percep_pred}
    The value of perception together with prediction is concave in $P_X$
    for a fixed $P_{Y|X}$.
\end{proposition}
\begin{proof}
The value under consideration is $$\E[\ell(X,Y,u^*)] - \E[\ell(X,Y,\psi^*(X))].$$
From Lemma~\ref{lm:R(X,Y)} and \ref{lm:R(X,Y|X)}, we know that with a fixed $P_{Y|X}$, the first term is concave in $P_X$ and the second term is linear in $P_X$, thus the value under consideration is concave in $P_X$.
\end{proof}
This result echoes the information-theoretic realization of the value of perception together with prediction as $H(X)$, which is concave in $P_X$. This value can thus be seen as a generalized form of entropy of $X$, implying that the uncertainty about $X$ in decision making can be fully resolved by observing it and making prediction on $Y$ based on it.

\begin{proposition}\label{prop:vo_commun}
    The value of communication have the following properties:
\begin{enumerate}[leftmargin=*]
    \item 
    For a fixed $P_{X}$, it is concave in $P_{Y|X}$.
    \item 
    For a fixed $P_{Y|X}$, it is linear in $P_X$.
\end{enumerate}
\end{proposition}
\begin{proof}
The value under consideration is $$\E[\ell(X,Y,\psi^*(X))] - \E[\ell(X,Y,\psi^{**}(X,Y))].$$ The claim is a consequence of Lemma~\ref{lm:R(X,Y|X)} and Lemma~\ref{lm:R(X,Y|X,Y)}.
\end{proof}
This result echoes the information-theoretic realization of the value of communication as $H(Y|X)$, which has the same properties. It shows that this value can be seen as a generalized form of conditional entropy of $Y$ given $X$, implying that the residue uncertainty of $Y$ after predicting it based on $X$ in decision making can be fully resolved by communication.

\begin{proposition}\label{prop:vo_coms_pred}
For a fixed pair of $P_X$ and $P_Y$, and the set of probability transition kernels $\mathcal{M}=\{P_{Y|X}:P_{Y|X}\circ P_X = P_Y\}$ that couple the pair, the value of common sense and the value of prediction both are convex in $P_{Y|X}$ on $\mathcal{M}$.
\end{proposition}
\begin{proof}
First note that the set $\mathcal{M}$ is a linear space, as $(\lambda P_{Y|X} + (1-\lambda)Q_{Y|X})\circ P_{X} = P_Y$ for all $P_{Y|X}, Q_{Y|X}\in \mathcal{M}$ and $\lambda\in[0,1]$.
In addition, the set of joint distributions induced by $P_X$ and $\mathcal{M}$, denoted by $\Pi = \{P_X P_{Y|X}: P_{Y|X}\in\mathcal{M}\}$ is also a linear space, as the marginal distributions of $\lambda P_{X,Y} + (1-\lambda)Q_{X,Y}$ are $P_X$ and $P_Y$ for all $P_{X,Y}, Q_{X,Y}\in {\Pi}$ and $\lambda\in[0,1]$.
It is thus meaningful to discuss convex functions defined on these sets.

The value of common sense can be written as
\begin{align}
\E[\ell(X,Y,\bar u)] - \E[\ell(X,Y,u^*)]  
=   D_{\ell}(P_{X,Y}, P_X P_Y) 
\end{align}
where we have used the definitions of $\bar u$ and $u^*$ and the decision-theoretic divergence in Lemma~\ref{lm:D_ell}.
From the property of the divergence stated in Lemma~\ref{lm:D_ell}, we know that with $P_X$ and $P_Y$ fixed, the value of common sense is convex in $P_{X,Y}$ on $\Pi$; as $P_{X,Y}$ linearly depends on $P_{Y|X}$ when $P_X$ is fixed, the value of common sense is also convex in $P_{Y|X}$ on $\mathcal M$.

The value of prediction can be written as
\begin{align}
& \E[\ell(X,Y,\bar\psi(X))] - \E[\ell(X,Y,\psi^*(X))]  \nonumber \\
= & \sum_{x\in\sX} P_X(x)  \big( \E[\ell(x,Y,\bar\psi(x)|X=x] - \E[\ell(x,Y,\psi^*(x)|X=x] \big) \\
= & \sum_{x\in\sX} P_X(x)  D_{\ell,x}(P_{Y|X=x}, P_Y) 
\end{align}
where in the last step we have used the the fact that $\bar \psi(x) = \argmin_{u\in\sU}  \E[\ell(x,Y,u)]$ and $\psi^*(x) = \argmin_{u\in\sU}  \E[\ell(x,Y,u)|X=x]$ and the definition of the decision-theoretic divergence.
From the property of the divergence stated in Lemma~\ref{lm:D_ell}, we know that with a fixed $P_Y$, $D_{\ell,x}(P_{Y|X=x}, P_Y)$ is convex in $P_{Y|X=x}$ for all $x\in\sX$.
With $P_X$ fixed as well, the value of prediction becomes a nonnegative linear combination of convex functions in $P_{Y|X}$ on $\mathcal{M}$, thus is convex in $P_{Y|X}$ on $\mathcal{M}$.
\end{proof}
This result echoes the information-theoretic realization of the values of common sense and prediction as $I(X;Y)$, which has a stronger property that it is convex in $P_{Y|X}$ as long as $P_X$ is fixed.
It shows that these two values both can be seen as a generalized form of mutual information between $X$ and $Y$, implying that the uncertainty of $Y$ in decision making can be partially resolved by using its statistical dependence on $X$.

\section{Extension to multi-agent settings and applications}
We can extend the definitions of the values of different operations in decision making to problems involving one decision maker and multiple agents. The extended results may be used to address questions on the importance and the best order of perception and prediction of the agents in such decision-making problems.

Consider the problem where the decision maker needs to take an action $u\in\sU$ to minimize the expected loss $\E[\ell(X^n, Y^n, u)]$ influenced by $n$ agents. The $i$th agent has an observable state $X_i\in\sX$ which could be their current status, and an unobservable state $Y_i\in\sY$ which could be their action to take or their future status, for $i=1,\ldots,n$. There is a joint distribution $P_{X^n,Y^n}$ of $X^n = (X_1,\ldots,X_n)$ and $Y^n = (Y_1,\ldots, Y_n)$, known to the decision maker. The action of the decision maker can be taken based on $P_{X^n,Y^n}$ and $X^n$.
Using the concepts developed in the previous sections, we may gauge the importance of an agent by examining the value of observing and making prediction based on that agent. For the $i$th agent, this value can be defined as
\begin{align}
\E[\ell(X^n, Y^n, u^*)] - \E[\ell(X^n, Y^n, \psi_i^*(X_i))]
\end{align}
where we redefine $u^*$ to be 
$
u^* \deq \argmin_{u\in\sU} \E[\ell(X^n, Y^n, u)] 
$,
and define 
$
\psi_i^* \deq \argmin_{\psi:\sX\rightarrow\sU}\E[\ell(X^n, Y^n, \psi(X_i))] 
$.
This value reflects the importance of taking the $i$th agent into account for decision making without or before taking any other agent into account. If only one agent is allowed to be chosen to be observed and used to predict other agents, the one with the greatest value defined above should be chosen.



Another way of defining the value or importance of the $i$th agent would be the reduction of risk by observing and making prediction based on this agent while the other agents are already taken into account:
$
\E[\ell(X^n, Y^n, \psi^{*n\setminus i}(X^{n\setminus i}))] - \E[\ell(X^n, Y^n, \psi^{*n}(X^n))]
$,
where 
$
\psi^{*n\setminus i} \deq \argmin_{\psi:\sX^{n-1} \rightarrow \sU}\E[\ell(X^n, Y^n, \psi(X^{n\setminus i}))] 
$
and 
$
\psi^{*n} \deq \argmin_{\psi:\sX^n\rightarrow\sU} \E[\ell(X^n, Y^n, \psi(X^n))] 
$
are the best policies when taking 
$X^{n\setminus i} \deq (X_1,\ldots, X_{i-1}, X_{i+1},\ldots,X_n)
$
and $X^n$ respectively into account for decision making.
As both definitions are viable for agent importance, we see that the importance of an agent essentially depends on its order in perception and prediction. In other words, it depends on which other agents have already been taken into account for decision making before observing and making predictions based on that agent.
This leads to an answer to the second question. 

The best order $(i_1,\ldots,i_n)$ to run perception and prediction on agents can be determined as
\begin{align}
i_1 = \argmin_{i\in[n]}  \E[\ell(X^n, Y^n, \psi_i^*(X_i))]
\end{align}
and for $k=2,\ldots,n$
\begin{align}
i_k = \argmin_{i\in[n]\setminus\{i_1,\ldots, i_{k-1}\}}  \E[\ell(X^n, Y^n, \psi^{*}_{i^{k-1},i}(X_{i_1},\ldots,X_{i_{k-1}}, X_i))]
\end{align}
where
\begin{align}
\psi^{*}_{i^{k-1},i} \deq \argmin_{\psi:\sX^k\rightarrow\sU} \E[\ell(X^n, Y^n, \psi(X_{i_1},\ldots,X_{i_{k-1}}, X_i))] .
\end{align}
Equivalently, $i_k$ maximizes among $i\in[n]\setminus\{i_1,\ldots, i_{k-1}\}$ the value of the $k$th agent to be taken into account for decision making:
\begin{align}
\E[\ell(X^n, Y^n, \psi^*_{i^{k-1}}(X_{i_1},\ldots,X_{i_{k-1}}))] 
- 
\E[\ell(X^n, Y^n, \psi^*_{i^{k-1},i}(X_{i_1},\ldots,X_{i_{k-1}}, X_i))] .
\end{align} 

Note that the above definitions and analysis only rely on the knowledge of the loss function and the joint distribution $P_{X^n,Y^n}$. With such knowledge, the importance and the best ordering of the agents for perception and prediction can be obtained in principle without making any observations, and the results can be used as ground truths for training machine learning models for estimating agent importance and prioritization in the absence of such knowledge, e.g.\ the machine learning based model in \cite{rank2tell}.
The method of analysis can also be extended to the situations where some observations are already made, e.g. by replacing $P_{X^n,Y^n}$ with $P_{X^{n}_{k+1},Y^n|X^k=x^{k}}$ if the observable states of the first $k$ agents have already been observed. 

\section{Conclusion}
In this work, we  decompose the operations needed for the decision-making procedure, formalize and examine the value of each operation, with an attempt to understand these operations in a coherent framework. Drawing granular parallels to information measures, both in concrete forms and in general math properties, 
the defined values  also provide new and deepened information-theoretic perspectives on decision making. We hope the work can help to answer practical questions such as quantifying the importance of paying attention to the agents in principled manners. 
We also look out to extensions of the definitions and analyses to dynamic or sequential decision makings. 

\section*{Acknowledgment}
The author would like to thank Behzad Dariush, Teruhisa Misu, Enna Sachdeva, Chenran Li and Christian Goerick for helpful discussions.
The title of the paper follows the one of a talk ``One if by Land, Two if by Sea'' given by Prof.\ Robert Gallager \cite{Gallager_talk_2015} on the occasion of Prof.\ Bruce Hajek's sixtieth birthday in October 2015 at Coordinated Science Laboratory, University of Illinois at Urbana-Champaign.

\appendix
\setcounter{lemma}{0}
\renewcommand{\thelemma}{\Alph{section}\arabic{lemma}}

\bibliography{vop}

@article{voi_human_ai,
  author  = {Ziyang Guo and Yifan Wu and Jason Hartline and Jessica Hullman},
  title   = {The Value of Information in Human-AI Decision-making},
  journal = {arxiv 2502.06152},
  year    = {2025}
}

@article{MALONEY20102362,
title = {Decision-theoretic models of visual perception and action},
journal = {Vision Research},
volume = {50},
number = {23},
pages = {2362-2374},
year = {2010},
note = {Vision Research Reviews},
issn = {0042-6989},
author = {Laurence T. Maloney and Hang Zhang},
}

@article{sims2003implications,
  author  = {Christopher A. Sims},
  title   = {Implications of Rational Inattention},
  journal = {Journal of Monetary Economics},
  year    = {2003},
  volume  = {50},
  number  = {3},
  pages   = {665--690},
  month   = apr,
  doi     = {10.1016/S0304-3932(03)00029-1}
}

@article{sims2006rational,
  author  = {Christopher A. Sims},
  title   = {Rational Inattention: Beyond the Linear-Quadratic Case},
  journal = {American Economic Review},
  year    = {2006},
  volume  = {96},
  number  = {2},
  pages   = {158--163},
  month   = may,
  doi     = {10.1257/000282806777212431}
}

@article{lindley1956measure,
  author  = {Lindley, D. V.},
  title   = {On a Measure of the Information Provided by an Experiment},
  journal = {Annals of Mathematical Statistics},
  year    = {1956},
  volume  = {27},
  number  = {4},
  pages   = {986--1005},
  month   = dec,
  doi     = {10.1214/aoms/1177728069}
}

@inproceedings{blackwell1951comparison,
  author    = {David Blackwell},
  title     = {Comparison of Experiments},
  booktitle = {Proceedings of the Second Berkeley Symposium on Mathematical Statistics and Probability},
  editor    = {Jerzy Neyman},
  year      = {1951},
  pages     = {93--102},
  publisher = {University of California Press},
  address   = {Berkeley, CA}
}

@article{blackwell1953equivalent,
  author  = {David Blackwell},
  title   = {Equivalent Comparisons of Experiments},
  journal = {Annals of Mathematical Statistics},
  year    = {1953},
  volume  = {24},
  number  = {2},
  pages   = {265--272},
}

@article{harsanyi1967i,
  author  = {Harsanyi, John C.},
  title   = {Games with Incomplete Information Played by {Bayesian} Players, {Part I}. The Basic Model},
  journal = {Management Science},
  year    = {1967},
  volume  = {14},
  number  = {3},
  pages   = {159--182}
}

@article{Herd21,
author = {Herd, S. and Krueger, K. and Nair, A.},
title = {Neural Mechanisms of Human Decision-Making},
journal = {Cognitive, Affective, and Behavioral Neuroscience},
volume = {21},
pages = {35-57},
year = {2021},
}

@article{Morelli22,
author = {Morelli, M. and Casagrande, M. and Forte, G. },
title = {Decision Making: a Theoretical Review},
journal = { Integrative Psychological and Behavioral Science},
volume = {56},
pages = {609-629},
year = {2022},
}

@article{BREHMER1992211,
author = {Berndt Brehmer},
title = {Dynamic decision making: Human control of complex systems},
journal = {Acta Psychologica},
volume = {81},
number = {3},
pages = {211-241},
year = {1992},
}

@book{it_book_pw,
author = {Yury Polyanskiy and Yihong Wu},
title = {Information Theory: From Coding to Learning},
publisher = {Cambridge University Press},
year = {2024},
}

@ARTICLE{Howard_voi,
  author={Howard, Ronald A.},
  journal={IEEE Transactions on Systems Science and Cybernetics}, 
  title={Information Value Theory}, 
  year={1966},
  volume={2},
  number={1},
  pages={22-26},
  doi={10.1109/TSSC.1966.300074}}

@article{DeGroot62,
author = {M. H. DeGroot},
title = {{Uncertainty, Information, and Sequential Experiments}},
volume = {33},
journal = {The Annals of Mathematical Statistics},
number = {2},
publisher = {Institute of Mathematical Statistics},
pages = {404 -- 419},
year = {1962},
}

@article{max_gen_ent,
author = {Peter D. Gr{\"u}nwald and A. Philip Dawid},
title = {{Game theory, maximum entropy, minimum discrepancy and robust Bayesian decision theory}},
volume = {32},
journal = {The Annals of Statistics},
number = {4},
publisher = {Institute of Mathematical Statistics},
pages = {1367 -- 1433},
year = {2004},
}

@ARTICLE{MER22,
  author={Xu, Aolin and Raginsky, Maxim},
  journal={IEEE Transactions on Information Theory}, 
  title={Minimum Excess Risk in {B}ayesian Learning}, 
  year={2022},
  volume={68},
  number={12},
  pages={7935-7955},
}

@INPROCEEDINGS{agent_priority19,
  author={Refaat, Khaled S. and Ding, Kai and Ponomareva, Natalia and Ross, Stéphane},
  booktitle={IEEE/RSJ International Conference on Intelligent Robots and Systems (IROS)}, 
  title={Agent Prioritization for Autonomous Navigation}, 
  year={2019},
}

@inproceedings{goal_important19,
author = {Gao, Mingfei and Tawari, Ashish and Martin, Sujitha},
title = {Goal-Oriented Object Importance Estimation in On-Road Driving Videos},
year = {2019},
booktitle = {International Conference on Robotics and Automation (ICRA)},
}

@inproceedings{important_semi22,
author = {Li, Jiachen and Gang, Haiming and Ma, Hengbo and Tomizuka, Masayoshi and Choi, Chiho},
title = {Important Object Identification with Semi-Supervised Learning for Autonomous Driving},
year = {2022},
booktitle = {International Conference on Robotics and Automation (ICRA)},
}

@article{Zhang2020InteractionGF,
  title={Interaction Graphs for Object Importance Estimation in On-road Driving Videos},
  author={Zehua Zhang and Ashish Tawari and Sujitha Martin and David J. Crandall},
  journal={IEEE International Conference on Robotics and Automation (ICRA)},
  year={2020},
}

@article{common_sense_cvpr23,
  title={Workshop on Machine Visual Common Sense: Perception, Prediction, Planning},
  author={Yining Hong and Bo Wu and Zhenfang Chen and Qinhong Zhou and Chuang Gan},
  journal={IEEE International Conference on Computer Vision and Pattern Recognition (CVPR)},
  year={2023}
}

@misc{Gallager_talk_2015,
  author = {Robert Gallager},
  title = {One if by land, Two if by sea},
  howpublished = {Presented at the workshop celeberating Prof. Bruce Hajek's sixtieth birthday at CSL, UIUC},
  year = {2015},
  month = {Oct.},
  url = {https://www.rle.mit.edu/rgallager/documents/HajekTalk.pdf},
}

@INPROCEEDINGS {rank2tell,
author = { Sachdeva, Enna and Agarwal, Nakul and Chundi, Suhas and Roelofs, Sean and Li, Jiachen and Kochenderfer, Mykel and Choi, Chiho and Dariush, Behzad },
booktitle = { 2024 IEEE/CVF Winter Conference on Applications of Computer Vision (WACV) },
title = {{ Rank2Tell: A Multimodal Driving Dataset for Joint Importance Ranking and Reasoning }},
year = {2024},
}

@article{CUV2025,
title = {A Mathematical Formulation of AGI in the (C, U, V) Framework},
journal = {Engineering},
year = {2025},
issn = {2095-8099},
author = {Di He and Cong Fang and Yisen Wang and Yujia Peng and Yizhou Wang and Song-Chun Zhu},
}

@article{AIAlignment,
author = {Ji, Jiaming and Qiu, Tianyi and Chen, Boyuan and Zhou, Jiayi and Zhang, Borong and Hong, Donghai and Lou, Hantao and Wang, Kaile and Duan, Yawen and He, Zhonghao and Vierling, Lukas and Zhang, Zhaowei and Zeng, Fanzhi and Dai, Juntao and Pan, Xuehai and Xu, Hua and O'Gara, Aidan and Ng, Kwan and Tse, Brian and Fu, Jie and Mcaleer, Stephen and Wang, Yanfeng and Yang, Mingchuan and Liu, Yunhuai and Wang, Yizhou and Zhu, Song-Chun and Guo, Yike and Yang, Yaodong and Gao, Wen},
title = {AI Alignment: A Contemporary Survey},
year = {2025},
publisher = {Association for Computing Machinery},
address = {New York, NY, USA},
volume = {58},
number = {5},
journal = {ACM Comput. Surv.},
}

\end{document}